\documentclass[twoside,leqno,twocolumn]{article}

\usepackage{amsmath}
\usepackage{hyperref}
\usepackage{cleveref}
\usepackage[letterpaper]{geometry}
\usepackage{ltexpprt}

\usepackage{graphicx}
\usepackage{placeins}

\newcommand{\mysection}{\section}
\newcommand{\mysubsection}{\subsection}

\usepackage{geometry}
\usepackage{subcaption}
\newcommand{\eps}{\varepsilon}
\newcommand{\ourcompactor}{linear compactor} % remember to also rename: "Combining KLL and linear compactors" and "Linear compactors"
\newcommand{\algorithmname}{\ourcompactor{} sketch}
\DeclareMathOperator{\poly}{poly}
\DeclareMathOperator{\E}{E}
\DeclareMathOperator{\err}{err}
\renewcommand{\O}{O}

\usepackage{xcolor}

\newcommand*\samethanks[1][\value{footnote}]{\footnotemark[#1]}
\newcommand{\beginproofwithqed}{\rm \trivlist \item[\hskip \labelsep{\it Proof.\/}]}
\newcommand{\myqedhere}{
\qquad\vbox{\hrule height0.6pt\hbox{%
   \vrule height1.3ex width0.6pt\hskip0.8ex
   \vrule width0.6pt}\hrule height0.6pt
  }
}

\begin{document}

\title{Learned Interpolation for Better Streaming Quantile Approximation with Worst-Case Guarantees}
\author{{Nicholas Schiefer}\thanks{Anthropic. Work done while at MIT.}
\and {Justin Y. Chen}\thanks{Massachusetts Institute of Technology}
\and {Piotr Indyk}\samethanks[2]
\and {Shyam Narayanan}\samethanks[2]
\and {Sandeep Silwal}\samethanks[2]
\and {Tal Wagner}\thanks{Amazon AWS. Work done while at Microsoft Research.}}

\date{}

\maketitle

\fancyfoot[R]{\scriptsize{Copyright \textcopyright\ 2023\\
Copyright for this paper is retained by authors}}

\begin{abstract} \small\baselineskip=9pt % required by SIAM
An $\eps$-approximate quantile sketch over a stream of $n$ inputs approximates the rank of any query point~$q$---that is, the number of input points less than $q$---up to an additive error of $\eps n$, generally with some probability of at least $1 - 1/\poly(n)$, while consuming $o(n)$ space.
While the celebrated KLL sketch of Karnin, Lang, and Liberty achieves a provably optimal quantile approximation algorithm over worst-case streams, the approximations it achieves in practice are often far from optimal.
Indeed, the most commonly used technique in practice is Dunning's t-digest, which often achieves much better approximations than KLL on real-world data but is known to have arbitrarily large errors in the worst case.
We apply interpolation techniques to the streaming quantiles problem to attempt to achieve better approximations on real-world data sets than KLL while maintaining similar guarantees in the worst case.	
\end{abstract}

\mysection{Introduction}

The quantile approximation problem is one of the most fundamental problems in the streaming computational model, and also one of the most important streaming problems in practice. 
Given a set of items~$x_1, x_2, \dots, x_n$ and a query point~$q$, the \emph{rank} of $q$, denoted $R(q)$, is the number of items in $\{x_i\}_{i = 1}^n$ such that $x_i \leq q$.
An $\eps$-approximate quantile sketch is a data structure that, given access to a single pass over the stream elements, can approximate the rank of all query points simultaneously with additive error at most $\eps n$.

Given its central importance, the streaming quantiles problem has been studied extensively by both theoreticians and practitioners.
Early work by Manku, Rajagopalan, and Lindsay \cite{manku} gave a randomized solution that used $\O((1/\eps) \log^2 (n \eps))$ space; their technique can also be straightforwardly adapted to a deterministic solution that achieves the same bound \cite{wang2013}.
Later, Greenwald and Khanna \cite{greenwald} developed a deterministic algorithm that requires only $\O((1/\eps) \log(n\eps))$ space.
More recently, Karnin, Lang, and Liberty~(KLL)~\cite{kll} developed the randomized KLL sketch that succeeds at all points with probability $1 - \delta$ and uses $\O((1/\eps) \log \log (1/\delta))$ space and gave a matching lower bound.

Meanwhile, streaming quantile estimation is of significant interest to practitioners in databases, computer systems, and data science who have studied the problem as well. Most notably, Dunning \cite{dunning} introduced the celebrated t-digest, a heuristic quantile estimation technique based on 1-dimensional $k$-means clustering that has seen adoption in numerous systems, including Influx, Apache Arrow, and Apache Spark.
Although t-digest achieves remarkable accuracy on many real-world data sets, it is known to have arbitrarily bad error in the worst case \cite{cormode}.

To illustrate this core tradeoff, Figure~\ref{fig:comparison} shows the rank function of the \texttt{books} dataset from the SOSD benchmark \cite{kipf, marcus}, along with KLL and t-digest approximations that use the same amount of space when the data set is randomly shuffled, and when the same data set is streamed in an adversarial order that we found to induce especially bad performance in t-digest.

\begin{figure*}[t]
\centering
\includegraphics[width=0.9\textwidth]{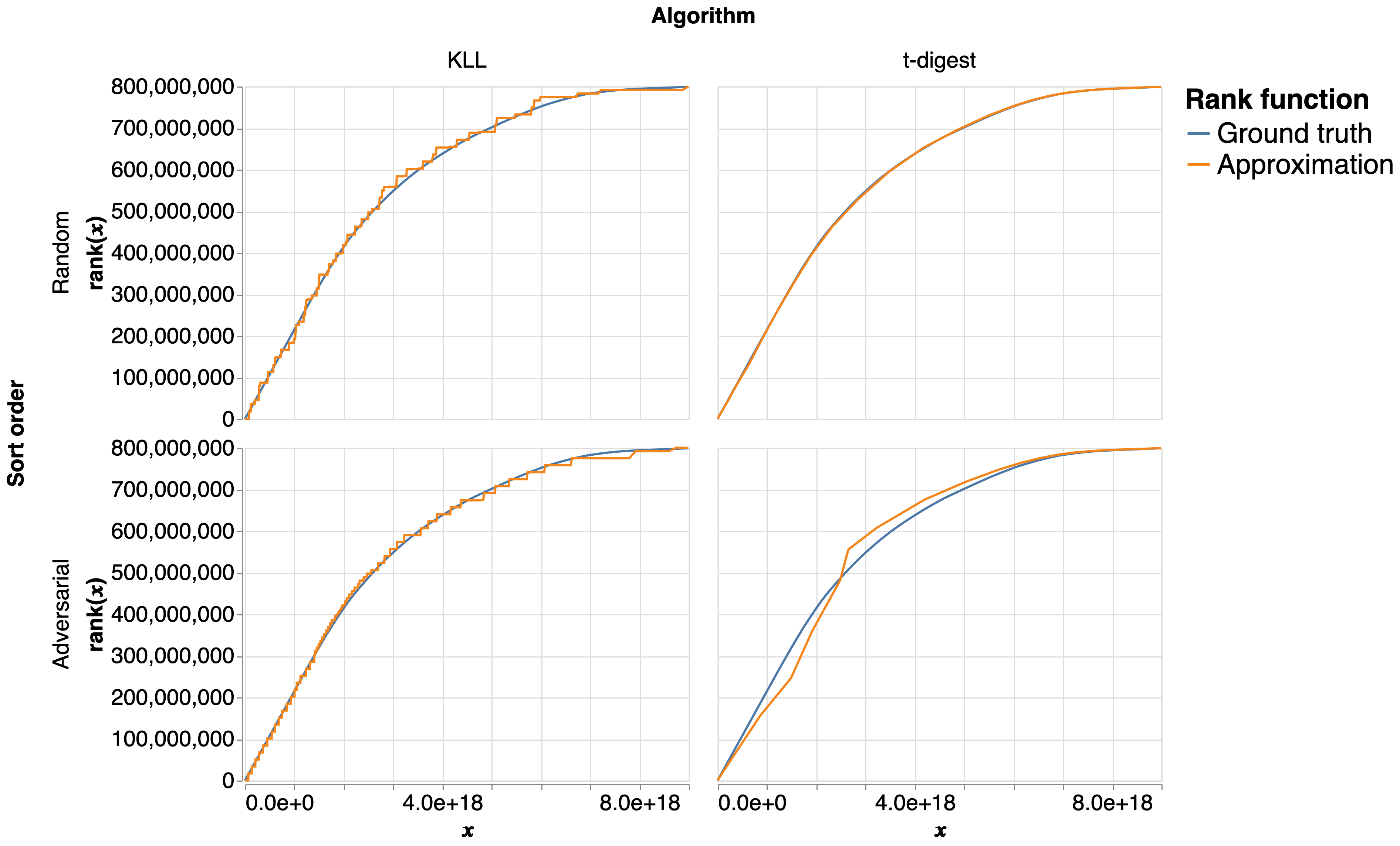}
\caption{Ground-truth and approximate rank functions for the SOSD \texttt{books} data set, with approximation by both the KLL and t-digest sketches. For the approximations, the data were presented in both randomly shuffled (top) and adversarial (bottom) order. In the adversarial case that we discovered, t-digest does much worse than KLL, demonstrating the value of worst-case correctness.}
\label{fig:comparison}

\end{figure*}

Recent advances in machine learning have led to the development of \emph{learning-augmented algorithms} which seek to improve solutions to classical algorithms problems by exploiting empirical properties of the input distribution \cite{mitzenmacher}.
Typically, a learning-augmented algorithm retains worst-case guarantees similar to those of classical algorithms while performing better on nicely structured inputs that appear in practical applications.
We might hope that a similar technique could be used for quantile estimation.

In fact, one of the seminal results in the field studied the related problem of \emph{learned index structures}. An index is a data structure that maps a query point to its rank.
Several model families have been tried for this learning problem, including neural networks and the successful recursive model index (RMI) that define a piecewise-linear approximation \cite{kraska}. 

Although learned indexes aim to answer rank queries, they do not solve the streaming quantiles estimation problem because they do not operate on the data in a stream. For example, training a neural network or fitting an RMI model require $\O(n)$ of the elements in the stream to be present in memory simultaneously, or require multiple passes over the stream.

\mysubsection{Our contributions.}

We present an algorithm for the streaming quantiles problem that achieves much lower error on real-world data sets than the KLL sketch while retaining similar worst case guarantees.
This algorithm, which we call the \algorithmname{}, uses \emph{linear interpolation} in place of parts of the KLL sketch.
Intuitively, this linear interpolation provides a better approximation to the true cumulative density function when that function is relatively smooth, a common property of CDFs of many real world datasets.

On the theoretical side, we prove that the \algorithmname{} achieves similar worst case error to the KLL sketch.
That is, the \algorithmname{} computes an $\eps$-approximation for the rank of a single item with probability $1 - \delta$ and space $O((1/\eps) \log^2 \log(1/\delta))$.
This is within a factor that is poly-log-logarithmic (in $1/\delta$) of the known lower bounds and the (rather complex) version of the KLL sketch that matches it \cite{kll}.
Our proof is a relatively straightforward modification of the analysis of the original KLL sketch, due to the general similarity of the algorithms.
In fact, we can view our algorithm as exploiting a place in the KLL sketch analysis that left some ``slack'' in the algorithm design.

In our experiments, we demonstrate that the \algorithmname{} achieves significantly lower error than the KLL sketch on a variety of benchmark data sets from the SOSD benchmark library \cite{kipf, marcus} and for a wide variety of input orders that induce bad behaviour in other algorithms like t-digest. 
In many cases, the \algorithmname{} achieves a space-error tradeoff that is competitive with t-digest, while also retaining worst-case guarantees.

\mysection{Understanding the KLL sketch}  \label{sec:kll}

The complete KLL sketch that achieves optimal space complexity is complex: it involves several different data structures, including a Greenwald-Khanna (GK) sketch that replaces the top $O(\log \log (1/\delta))$ compactors.
Here, we present a simpler version of the KLL sketch that uses $O((1/\eps) \log^2 \log (1/\delta))$ space---just a factor of $O(\log \log (1/\delta))$ away from optimal---and is commonly implemented in practice \cite{ivkin}, presented in Theorem 4 of~\cite{kll}.
In the remainder of this paper, we refer to this sketch as the \emph{non-GK KLL sketch}.

\mysubsection{The non-GK KLL sketch.}

The basic KLL sketch is composed of a hierarchy of \emph{compactors}.
Each of the $H$ compactors has a capacity~$k$, which defines the number of items that it can store.
Each item is also associated with a (possibly implicit) \emph{weight} which represents the number of points from the input stream that it represents in the sketch.
All points in the same compactor have the same weight.

When a compactor reaches its capacity, it is compacted.
A compaction begins by sorting the items.
Then, either the even or odd elements in the compactor are chosen, and the unchosen items are discarded.
The choice to discard the even or odd items is made with equal probability.
The chosen items are then placed into the next compactor in the hierarchy and the points are all assigned a weight twice what they began with.
This general setup is common to many streaming quantiles sketches \cite{manku, kll}.

To predict the rank of a query point~$q$, we return the sum of the weights of all points, in all compactors, that are at most $q$.

A key contribution of the KLL sketch is to use different capacities for different compactors.
We say that the first compactor where points arrive from the stream has a height of 0, and each successive compactor has a height one higher than the compactor below it, so that the top compactor has height~$H - 1$.
In KLL, the compactor at height $h$ has capacity $k_h = \max(k c^{H - h}, 2)$, where $k$ is a space parameter that defines the capacity of the highest compactor and $c$ is a scale parameter that is generally set as $c = 2/3$.

\mysubsection{Analysis of the non-GK KLL sketch.}

Here, we give a somewhat simplified---to focus on the essential details---version of the analysis of the non-GK KLL sketch.
Consider the non-GK KLL sketch described above that terminates with $H$ different compactors.
The weight of the items at height $h$ is $w_h = 2^h$.
Let $m_h$ be the number of compaction operations in the compactor at height~$h$.

Consider a single compaction operation in the compactor at height~$h$ and a point $x$ in that compactor at that time.
If $x$ was one of the even elements in the compactor, the total weight to the left of it, which defines its rank, is unchanged by the compaction.
If $x$ is one of the odd elements in the compactor, the total weight either increases by $w_h$ (if the odd items are chosen) or decreases by $w_h$ (if the even items are chosen).
For the $i$th compaction operation at level $h$, let $X_{i, h}$ be $1$ if the odd items were chosen and $-1$ if the even items were chosen.
Observe that $\E[X_{i, h}] = 0$ and $|X_{i, h}| \leq 1$.
Then the total error introduced by all compactions at level $h$ is
$
\sum_{i = 1}^{m_h} w_h X_{i, h}
$.
Consider any point $x$ in the stream. The error in $R(x)$ introduced by compaction at all levels up to a fixed level~$H'$ is therefore
$
\sum_{i = 0}^{H' - 1} \sum_{i = 1}^{m_h} w_h X_{i, h}
$.

Applying a two-tailed Hoeffding bound to this error, we obtain that
\begin{align*}
& \Pr[\text{error is $> \eps n$}] \\
&= \Pr\left[\left|\sum_{i = 0}^{H' - 1} \sum_{i = 1}^{m_h} w_h X_{i, h}\right| > \eps n \right] \\
&\leq
2 \exp \left( - \frac{\eps^2 n^2}{2 \sum_{i = 0}^{H' - 1} \sum_{i = 1}^{m_h} w_h^2} \right).
\end{align*}

This addresses the error introduced by all layers up to $H'$.
Notice that if we set $H' = H$, then the error bound is dominated by the weight terms from the highest compactors.
To get around this, the non-GK KLL sketch sets the capacity of the final $s = \O(\log \log (1/\delta))$ compactors to a fixed constant~$k$ and analyzes them separately: it is assumed to contribute its worst possible error of $w_h$ for reach compaction.
This is the key lemma in the KLL analysis and the point of departure for the \algorithmname{}.

\mysection{The \algorithmname} \label{sec:algorithm}

We propose a streaming quantile approximation algorithm that combines our empirical and theoretical observations about how KLL might be improved.
We leave the basic architecture of the non-GK KLL sketch unchanged.
Like the optimal KLL sketch, which replaces the top $O(\log \log (1/\delta))$ compactors with a Greenwald-Khanna sketch, we replace some of these top compactors with another data structure.
In our case, we replace the top $t = O(1)$ compactors with a structure that we call a \emph{\ourcompactor{}}.

\paragraph{Linear compactors.}
A \emph{\ourcompactor{}} is a sorted list of elements, each of which is a pair of an item from the stream and a weight.
As in KLL, the weight represents the number of stream items that the item represents; unlike in KLL, this weight varies between elements in the list and may be an arbitrary floating point number, rather than a power of two.
Like a KLL compactor, a \ourcompactor{} has a capacity which we fix to $tk$, the total capacity of the (fixed-size) compactors it replaces.
When that capacity is exceeded, it undergoes \emph{compaction} and only half of its elements are retained.

A KLL compactor $C_h$ at height $h$ implicitly represents a piecewise-constant function~$f$: specifically,
\[
f(q) = \sum_{x \in C_h : x \leq q} w_h.
\]
This function is the contribution of this compactor to the approximated rank of a query point~$q$.
A \ourcompactor{} implicitly represents a \emph{piecewise-linear} function which also contributes to the rank of $q$.
Given a \ourcompactor{} $L = \{(y_1, w_1), (y_2, w_2), \dots, (y_k, w_k)\}$ with $y_1 \leq y_2 \leq \cdots \leq y_k$, the contribution of $L$ to the the rank of $q$ is
\begin{equation}
\label{eq:rank}
f_L(q) = \underbrace{\sum_{i = 1}^{i^\ast - 1} w_i}_{\text{KLL-style term}} +\quad \underbrace{w_{i^\ast} \frac{q - y_{i^\ast - 1}}{y_{i^\ast} - y_{i^\ast - 1}}}_{\text{interpolation term}}
\end{equation}
where $i^\ast$ is the smallest index such that $y_{i^\ast} > q$.
In effect, we spread the weight of $y_{i^\ast}$ over the entire interval between $y_{i^\ast - 1}$ and $y_{i^\ast}$, with uniform density, rather than treating it as a point mass at $y_{i^\ast}$ exactly.
The resulting contribution $f_L(q)$ is a monotone, piecewise-linear function, as desired.

\paragraph{Adding points to a \ourcompactor{}.}

Our \ourcompactor{} receives points from the last of the KLL-style compactors, each with a fixed weight of $w_{H - t - 1}$.
These points and weights cannot be merged by merely concatenating the arrays.
To see this, consider adding a single point~$b$ with unit weight to a compactor with two points~$a$ and $c$ with $a < b < c$, and where $c$ has weight $w$.
The weight of $c$ after the compaction should not be $w$ since the weight of $c$ before the addition should be spread uniformly over the entire interval $[a, c]$.

Instead, we add a set of new points~$y_1 < y_2 < \cdots < y_m$ to an existing set of points~$x_1 < x_2 < \cdots < x_n$ by merging the two lists of points into one list and sorting them into the list~$z_1 < z_2 < \cdots < z_{m + n}$.
Next, we set $w(z_1)$ equal to the weight of $z_1$ in the original list and compute the new weights recursively.
Assuming that $z_i = x_i$ without loss of generality, we set
\[
w(z_i) = w(x_i) \frac{x_i - z_{i - 1}}{x_i - x_{i -1 }} + w(y_\ast) \frac{x_i - z_{i - 1}}{y_\ast - y_{\ast - 1}}
\]
where $y_\ast$ is the first $y_i$ such that $y_i > z_i$.

Equivalently, we convert each of the weight functions into a rank function using Equation~\ref{eq:rank}, sum those, and then compute the finite differences to obtain the final weight function.

\paragraph{Compacting a \ourcompactor{}.}
 
Lastly, we describe the process for compacting a \ourcompactor{}.
Given a parameter $\alpha \in [0, 1]$ and a \ourcompactor{}~$C$ containing $n$~points, we wish to obtain a new \ourcompactor{}~$C'$ with $\alpha n$~points with the following properties:
\begin{itemize}
	\item The points in $C'$ are subset of the points in $C$.
	\item The total weight of the points in both compactors is the same, so that $\sum_{x \in C} w(x) = \sum_{x \in C'} w(x')$.
	\item For every point~$x \in C'$, the rank $f_C(x) = f_{C'}(x)$.
	\item The ``error" introduced by the compaction is as small as possible. That is, for some loss function~$L$, we would like $\sum_{x \in C} L(f_{C'}(x), f_C(x))$ to be as small as possible.
\end{itemize}

In this paper, we use $\alpha = 1/2$, although in principle other values could be used.

It is important that this procedure can be completed efficiently.
In our experiments, we primarily use supremum ($\ell_\infty$) loss~$L(x, x') = \sup_x |x - x'|$.
This can be minimized using a dynamic programming technique introduced by \cite{jagadish}.

\mysection{Analysis}

We give a worst-case analysis of our algorithm that matches the worst-case analysis for the version of the non-GK KLL sketch:

\begin{theorem} \label{thm:main}
The \algorithmname{} described in \Cref{sec:algorithm} computes an $\eps$-approximation for the rank of a single item with probability $1 - \delta$ with space complexity $\O((1/\eps) \log^2 \log (1/\delta))$.
\end{theorem}

Our technique analyzes the error introduced by each compactor, using two techniques.
To analyze the error of the KLL-style compactors of the \algorithmname{}, we prove that they introduce precisely the same error as they would in a non-GK KLL sketch run on the same stream.
We then apply the two-part analysis of the non-GK KLL sketch, analyzing the first $H - s$ compactors and the $(H - s)$th through $(H - s + t)$th compactors separately. To analyze the error of the \ourcompactor{} at the top, we analyze the error introduced per compaction.
We then analyze the number of compactions of the \ourcompactor{} and therefore the total error introduced by the \ourcompactor{}.

Consider a stream $X = x_1, x_2, \dots, x_n$.
Let $S(X)$ be a non-GK KLL sketch computed on this stream that terminates with $H$ compactors and let $S_b(X)$ be the $b$th compactor of $S(x)$.
Similarly, let $S'(X)$ be a \algorithmname{} computed on this stream with $H - t$ levels of KLL-style compactors and one \ourcompactor{} at level $H - t + 1$.
Let $S'_b(X)$ be the $b$th compactor $S'(X)$.

Following \cite{kll}, let $R(S, x, h)$ be the rank of item~$x$ among all points in compactors in the sketch~$S$ at heights at most $h' \leq h$ at the end of the stream.
For convenience, we set $R(x, 0)$ to be the true rank of $x$ in the input stream.
Let $\err(S, x, h) = R(S, x, h) - R(S, x, h - 1)$ be the total change in the approximate rank of $x$ due to the compactor at level $h$.
The total error decomposes into this error per compactor as $\sup_x |R(x, 0) - S'(x)| = \sum_{h = 1}^H \err(S', x, h)$.

\paragraph{Analyzing the KLL compactors.}
In both $S$ and $S'$, stream elements only move from lower compactors to higher ones, and the compactor at level~$b$ at any point while processing the stream is defined entirely by the compactors at \emph{lower} levels up to that point.
Therefore, for all $b < H - t$, $S'_b(X) = S(X)$.

In a KLL sketch, the lowest compactors all have a capacity of exactly~2.
As the authors note, a sequence of $H''$ compactors that all have capacity~2 is essentially a sampler: out of every $2^{H''}$ elements they select one uniformly and output it with weight $2^{H''}$.
This means that these compactors---in both KLL and \algorithmname{}---can be implemented in $O(1)$ space.
 
To handle the other KLL compactors, we use a theorem from \cite{kll} as a key lemma:

\begin{theorem}[Theorem 3 in \cite{kll}] \label{thm:kll-analysis}

Consider the non-GK KLL sketch $S(X)$ with height $H$, and where the compactor at level $h$ has capacity $k_h \geq k c^{H - h}$. Let $H''$ be the height at which the compactors have size greater than 2 (i.e., where the compactors do not just perform sampling). For any $H' > H''$, we have
\begin{align*}
\Pr\left[\sum_{h = 1}^H \err(S, x,h)> 2 \eps n\right] &\leq 2 \exp \left( -c \eps^2 k 2^{H - H''} / 32 \right)\\  &+ 2 \exp \left( -C \eps^2 k^2 2^{2(H - H')} \right).
\end{align*}
\end{theorem}

\paragraph{Analyzing the \ourcompactor{}.}

As mentioned, we will analyze the error introduced by the \ourcompactor{} compaction-by-compaction.
Specifically, we analyze the \algorithmname{} between the end of one compaction and the end of the following compaction.
During this interval, a total of $d$ items of weight $2^{H - t}$ are added to the \ourcompactor{}, where either $d = tk$ if the \ourcompactor{} has never compacted or $d = tk/2$ if it has.

Let $f$ be the piecewise linear rank function of the full \ourcompactor{} right before the compaction with endpoint set~$Z$ comprising $z_1 < z_2 < \cdots < z_{tk}$ and weight function $w$.
Let $f'$ be the piecewise linear rank function of the \ourcompactor{} immediately after the compaction, with endpoint set~$Z' \subset Z$, weight function $w'$, and $|Z'| = |Z|/2$. 

The \ourcompactor{} compaction procedure removes some of the items in the \ourcompactor{}.
A \emph{run} is a sequence of removed elements that are adjacent in sorted order.
We show that the error introduced by a \ourcompactor{} is bounded by the greatest run of displaced weight.

\begin{lemma} \label{lma:displaced-weight}
Organize $Z \setminus Z'$ into continuous runs of adjacent removed elements, and let $F_i$ be the total weight of the $i$th run.
Then $\sup_{z \in Z} |f(z) - f'(z)| \leq \max_i F_i$.	
\end{lemma}

\beginproofwithqed{}
Fix a run with endpoints $a$ and $b$ and let its total weight be $F = \sum_{i = a + 1}^{b - 1} w(z_i)$.
Consider any point $z_j$ in that run, so that $a < j < b$.
Its original rank was $f(z_j) = \sum_{i = 1}^j w(z_i)$ while its new rank is, by construction, $f'(z_j) = \sum_{i = 1}^a w(z_i) + \frac{F + w(z_b)}{z_b - z_a} (z_j - z_a)$.
Therefore,
\begin{align*}
|f(z_j) - f'(z_j)|
&= \left| \sum_{i = a + 1}^j w (z_i) - \frac{F + w(z_b)}{z_b - z_a} (z_j - z_a) \right| \\
&= \left| \sum_{i = a + 1}^j w (z_i) - \sum_{i = a + 1}^b w(z_i) \frac{z_j - z_a}{z_b - z_a}  \right| \\
&\leq \sum_{i = a + 1}^{b - 1} w (z_i) \\
&= F.
\myqedhere
  \end{align*}

Next, we show that the greater error introduced by a linear compaction step occurs \emph{at one of the discarded endpoints}:

\begin{lemma} \label{lma:discarded-endpoints}
There is some $z_i \in Z \setminus Z'$ such that $\sup_{x \in [z_1, z_{tk}]} |f(x) - f'(x)| = |f(z_i) - f'(z_i)|$.
\end{lemma}

\begin{proof}
Consider any point $x \in [z_1, z_d]$.
If $x$ is one of the endpoints retained after compaction $z_j \in Z'$, then by construction $f(x) = \sum_{i \leq j} w(z_i) = f'(x)$.
Our claim does not depend on the error if $x$ is one of the endpoints in $Z \setminus Z'$.

Suppose then that $x$ is not in the original endpoint set~$Z$.
Let $z_a$ and $z_b$ be the left and right neighbours of $x$ in $Z$.
By the definition of the \ourcompactor{},
\begin{align*}
f(z_a) &= \sum_{i = 1}^a w(z_i), \\
f(x) &= \sum_{i = 1}^a w(z_i) + \frac{w(z_b)}{z_b - z_a} (x - z_a),\\
f(z_b) &= \sum_{i = 1}^b w(z_b).
\end{align*}
Let $z_{a'}$ and $z_{b'}$ be the left and right neighbours of $x$ in $Z'$.
By definition, the weight $W := w'(z_{b'}) = \sum_{i = a' + 1}^{b'} w(z_i)$ and so we have
\[
f'(x) = \sum_{i = 1}^{a'} w(z_i) + \frac{W}{z_{b'} - z_{a'}} (x - z_{a'}).
\]
Therefore,
\begin{align*}
f(x) - f(x') = &\sum_{i = a' + 1}^{a} w(z_i) + \\
&\left( \frac{w(z_b)}{z_b - z_a} - \frac{W}{z_{b'} - z_{a'}}\right) (x - z_{a'}).
\end{align*}

Observe that this expression obtains its extremum on the interval~$[z_a, z_b] \subset [z_{a'}, z_{b'}]$ at either $z_a$ or $z_b$, depending on the sign of $D = \frac{w(z_b)}{z_b - z_a} - \frac{W}{z_{b'} - z_{a'}}$. In either case, $|f(x) - f(x')|$ achieve its maximum at one of the endpoints~$z_a$ or $z_b$, completing the proof.
\end{proof}

We use a simple counting argument to bound the size of the majority of the weights in a \ourcompactor{}:

\begin{lemma} \label{lma:compactor-weight-bound}
Consider a \ourcompactor{} that has just completed its $c$th compaction.
At least half of the endpoints~$Z$ in the \ourcompactor{} have weight at most $(2c + 3) 2^{H - t}$.
\end{lemma}

\begin{proof}
Every point enters the \ourcompactor{} with weight $2^{H - t}$.
After $c$ compactions, a total of $(2 + c)tk/2$ such points have entered the compactor.
A compaction operation conserves the total weight of points so the total weight of the compactor is $(2 + c) 2^{H - t} tk/2$.

Suppose that more than half of the $tk$ points currently in the compactor have  weight at most $T$.
These points have a total weight greater than $T tk/4$ while the remaining point s each have weight at least $2^{H - t}$ and so have total weight at least $2^{H - t} tk/4$.
The total weight is therefore $(T + 2^{H - t}) tk / 4$.
This weight must not exceed the total conserved weight $(2 + c) 2^{H - t} tk/2$, and so we have
\[
\frac{(T + 2^{H - t}) tk}{4} \leq \frac{(2 + c) 2^{H - t} tk}{2}.
\]
Rearranging, we obtain that our result holds for any $T \leq (2c + 3) 2^{H - t}$.
\end{proof}

Combining these lemmas, we obtain a bound on the error introduced during a single compaction step.

\begin{theorem} \label{thm:single-linear-compaction}
Suppose that the compaction being studied is the $(c + 1)$th compaction. 
The error introduced during this compaction step is $\sup_x |f(x) - f'(x)| \leq  (c + 2) 2^{H - t + 1}$.
\end{theorem}
\beginproofwithqed
We construct a particular post-compaction distribution of weights as follows.
Let $f''$ be the rank function for that post-compaction state.
During this interval, there were $tk/2$ points with weight $2^{H - t}$ that we added to the \ourcompactor{} for the first time.

In addition, there were $tk/2$ points remaining from a previous linear compaction.
We sort the $tk/2$ new points and keep every fourth point, discarding the rest and reallocting their weight to the next highest retained point (of either type).
By \Cref{lma:compactor-weight-bound}, there exists at least $tk/4$ of the existing points in the \ourcompactor{} with weight at most $2^{H - t + c}$.
We sort these points and discard every other point.
In total, we discard the required $tk/2$ points.

Observe that the longest possible run in this compaction consists of one of the existing points and three (out of a sequence of four) of the new points that were discarded.
By \Cref{lma:displaced-weight}, the error introduced on any of the original endpoints by this compaction is bounded by the sum of the weights of the points in the run: in this case, that sum is $2^{H - t} + 3 \cdot (2c + 3) 2^{H - t} \leq (c + 2) 2^{H - t + 1}$.
By \Cref{lma:discarded-endpoints}, we find that the error introduced by $f''$ is $\sup_x |f(x) - f''(x)| \leq (c + 2) 2^{H - t + 1}$.

We have exhibited a particular feasible solution to the optimization problem in the linear compaction.
Our actual algorithm finds, among all such feasible solutions, the one that minimizes this error function; it follows that
\begin{align*}
\sup_x |f(x) - f'(x)| &\leq
\sup_x |f(x) - f''(x)| \\ &\leq (c + 2) 2^{H - t + 1}.
\myqedhere
\end{align*}

\paragraph{Combining KLL and linear compactors.}

Lastly, we combine our analysis of the KLL and \ourcompactor{} to obtain an overall error bound and prove \Cref{thm:main}.
Our analysis closely follows the form of the proof of Theorem 4 in \cite{kll}.

\begin{proof}[Proof of \Cref{thm:main}]
First, we analyze the compactors with height at most $H - t$, including the sampling compactors.
These are all KLL-style compactors; by \Cref{thm:kll-analysis} these compactors will contribute error at most $\eps n$ with probability $1 - \delta$ so long as $\eps k 2^s \geq c' \sqrt{\log (2/\delta)}$ for a sufficiently small $c'$.
Second, we analyze the top $s - t$ compactors. 
The error introduced by these compactors is bounded by the error of the equivalent non-GK KLL sketch where we have a full $s$ equal-size compactors at the top.
This error is in turn bounded by $\sum_{h = H - s + 1}^h m_h w_h = \sum_{h=H - s + 1}^H n/k = sn/k$, where $m_h$ is the number of times that the KLL compactor at level $h$ is compacted and $w_h = 2^h$ is the weight associated with that compactor; this is at most $\eps n$ so long as $s \leq k \eps$.
Taking  $k = O((1/\eps) \log \log (1/\delta))$ and $s = O(\log \log (1/\delta))$ as in KLL, we satisfy both of these conditions.

\begin{figure*}[t]
\includegraphics[width=0.9\textwidth]{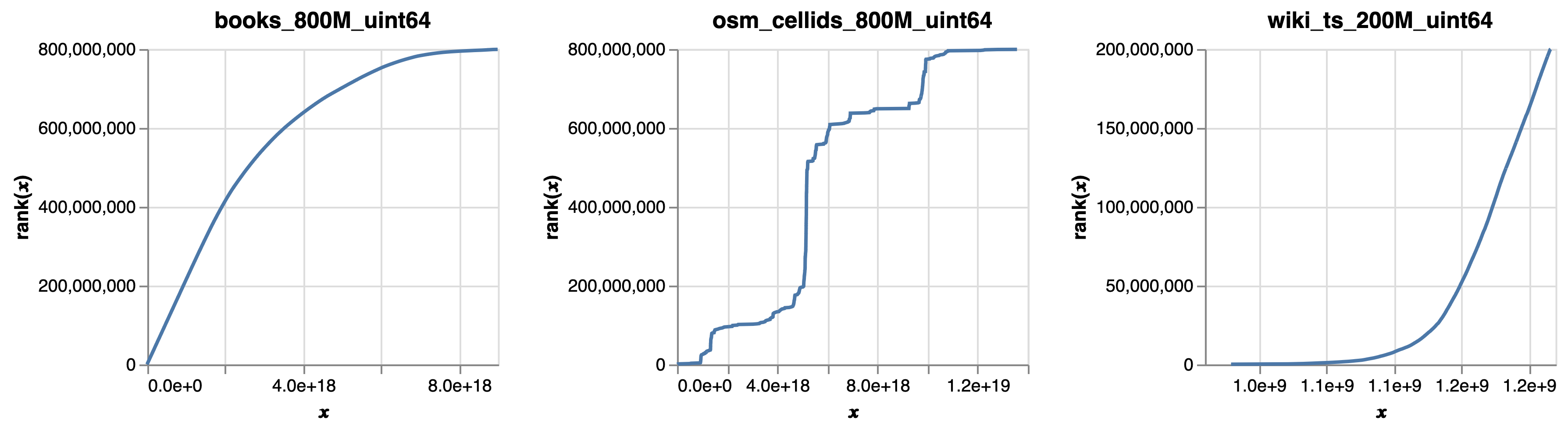}
\caption{The rank functions for the three SOSD data sets used in our experiments. The three data sets have rank functions with distinctive shapes, allowing us to compare the algorithms in a variety of settings.
}
\label{fig:sosd_cdf}	
\end{figure*}

Lastly, we analyze the single \ourcompactor{} with size~$tk$ that replaces the top $t < s$ KLL compactors.
Let $M$ be the number of compactions of the \ourcompactor{}.
Observe that since between each compaction of the \ourcompactor{} we add $tk/2$ entries, each with weight $2^{H - t}$ to the compactor, and so $M \leq 2n/(tk 2^{H - t})$.
Applying \Cref{thm:single-linear-compaction}, and summing the error introduced per compaction, the total error is
\begin{align*}
\sum_{c = 1}^M (c + 2) 2^{H - t + 1}
&= M 2^{H - t} + 2^{H - t} M(M + 1) \\
&\leq 2^{H - t + 1} M^2 \\
&\leq \frac{8 n^2}{t^2 k^2 2^{H - t}}.
\end{align*}

Our compactors are sized at each level in the same way as a non-GK KLL-sketch.
As in the KLL analysis, we have $H \leq \log(n/ck) + 2$ for a constant $0 < c < 1$.
Therefore, our error is bounded by
\[
\frac{8 n^2}{t^2 k^2 2^{\log(n/ck) + 2 - t}} \leq 
\frac{8 ck n^2}{t^2 k^2 n 2^{2 - t}} =
\frac{c n 2^{t + 1}}{t^2 k}.
\]For constant~$t$ and any $k = O((1/\eps) \log \log(1/\delta))$ as in KLL, this is at most $\eps n$.
Therefore, the total error of the sketch is $O(\eps n)$ as required.

Each part of the sketch contributes some space.
The KLL compactors increase geometrically in size, so the space used by the KLL portion of the sketch is dominated by the top $s - t$ compactors and uses $O(sk) = O((1/\eps) \log^2 \log (1/\delta))$ space.
The \algorithmname{} uses twice as much space per element as a KLL compactor, for a total of $O(tk) = O(k)$ space, so the total space usage is $O((1/\eps) \log^2 \log (1/\delta))$.
\end{proof}

\mysection{Experiments}

We wrote a performant implementation of our algorithm and evaluated its empirical error over a wide range of space parameters~$k$ and several \ourcompactor{} heights~$t$.
Our experiments were conducted on the recent SOSD benchmarking suite \cite{kipf, marcus} for learned index structures.
Each SOSD benchmark consists of a large number (generally 200 to 800 million) of 64-bit unsigned integer values.
Of particular interest were the \texttt{books}, \texttt{osm\_cellid}, and \texttt{wiki\_ts} data sets, since the rank functions of these three data sets have distinctly different shapes, as shown in \Cref{fig:sosd_cdf}.

\paragraph{Parameterization.}
The algorithm is parameterized by the KLL space parameter~$k$, which determines the size of the largest compactors and the \ourcompactor{}, and $t$, the number of KLL compactors that are replaced by the \ourcompactor{}.
Our worst-case bound holds for any constant~$t$ but this bound is exponential in $t$.
In practice, we experimented with a variety of small but non-zero values ($t = 1, 2, 3$).
We see $t$ as a parameter that is tunable based on the desired empirical performance and desired worst-case guarantees and expect that it will be selected appropriately on an application-by-application basis.

\paragraph{Implementation details.}
We implemented our algorithms in C++ with Python bindings for experiment management and data analysis.
Our implementation is reasonably performant: in informal experiments, it achieves a throughput that is only about three times less than that of highly-optimized,  production-quality KLL implementations.
This performant implementation allowed us to work with the entirety of the SOSD data sets; in our preliminary work, we found that many promising algorithms would only show improvements over KLL on moderately-sized data sets of less than a million points.
Our implementation supports any integer~$t \geq 0$: when $t = 0$, our implementation is identical to the commonly implement variant of KLL without the Greenwald-Khanna sketch.

\paragraph{Baselines.} Our algorithm is most naturally compared to KLL since the KLL sketch can be seen as an instance of the \algorithmname{} with no \ourcompactor{}.
We ran our experiments on our implementation of (non-GK) KLL (by setting $t = 0$) and validated those results with an open-source implementation from Facebook's Folly library \cite{folly}.
Like most implemented version of the KLL sketch, neither of these include the final Greenwald-Khanna sketch that is required to achieve space-optimality.

In addition to the non-GK KLL sketch, which offers worst-case guarantees, we ran experiments on the t-digest \cite{dunning}, which is commonly used in practice but is known to have arbitrarily bad worst-case performance \cite{cormode}.
We used the C++ implementation of t-digest in the \texttt{digestible} library \cite{digestible}.

\paragraph{Stream order.}
We found that many streaming quantile approximation algorithms without worst-case guarantees achieve very low error compared to the KLL sketch if they are given an input stream in a particular order but high error on other input orders.
For example, Figure~\ref{fig:comparison} shows that, even for a fixed set of inputs with a smooth rank function (\texttt{books}), there exists an adversarial order that makes the t-digest approximation have high error.
This observation might be of independent interest.

We evaluated the \algorithmname{} and the baselines on a variety of input orders for each data set:
\begin{itemize}
\item Random: the data are shuffled with a fixed seed.
\item Sorted: the data are presented in a sorted order.
\item First half sorted, second half reverse-sorted: the first half of the stream has the first half of the sorted data, in that order. The second half of the stream has the second half of the sorted data presented in \emph{reverse-sorted order}.
\item Flip flop: the stream has the smallest element, then the largest element, then the second-smallest element, then the second-largest element, and so on. This is the adversarial order from Figure~\ref{fig:comparison}.
\end{itemize}

\mysubsection{Experiment results and discussion.}

Our primary tool for insight into our experiments is the space-error tradeoff curve that shows how the total  space needed for the sketch compares to the empirical error between the exact rank function and the approximation defined by the sketch.
We obtain these curves for three different data sets from SOSD, four different sort orders, and  four different algorithms; these curves are shown in \Cref{fig:space-error-tradeoff}.
We use average L1 error, defined for a data set $X$ as $\sum_{x \in X} |f(x) - f'(x)|$.
Qualitatively, the \algorithmname{} is never significantly worse than KLL, even on our adversarial input orders like flip flop, and is often competitive with---or even better than---t-digest.
The differences are most pronounced on the \texttt{books} dataset, which has a smooth CDF that is extremely well-approximated by the \ourcompactor{} sketch's piecewise linear representation.

For a more quantitative understanding of the performance of the \algorithmname{} compared to KLL and t-digest, we produced diagrams of the ``possible error ratio hulls'', shown in Figure~\ref{fig:error-ratios}.
To obtain such a hull, we first determined the upper and lower frontiers for the data points (in Figure~\ref{fig:space-error-tradeoff}) for each algorithm.
These frontiers form an ``envelope'' or hull that encompasses all of the points for each dataset: an example of such a hull is shown in Figure~\ref{fig:hull}.
We then interpolate the envelope to obtain smooth curves (as in Figure~\ref{fig:hull}) and compute the ratios with respect to a another algorithm's envelope (between the upper/lower and lower/upper pairs), producing a hull that shows the range of behaviour between the ``worst case'' and ``best case'' performance of  the two algorithms.
We see that the \algorithmname{} achieves an error that is between $3 \times$ worse and $10\times$ better than KLL and between $10\times$ worse and $20\times$ better than t-digest.

\begin{figure}[t]
\includegraphics[width=0.95\columnwidth]{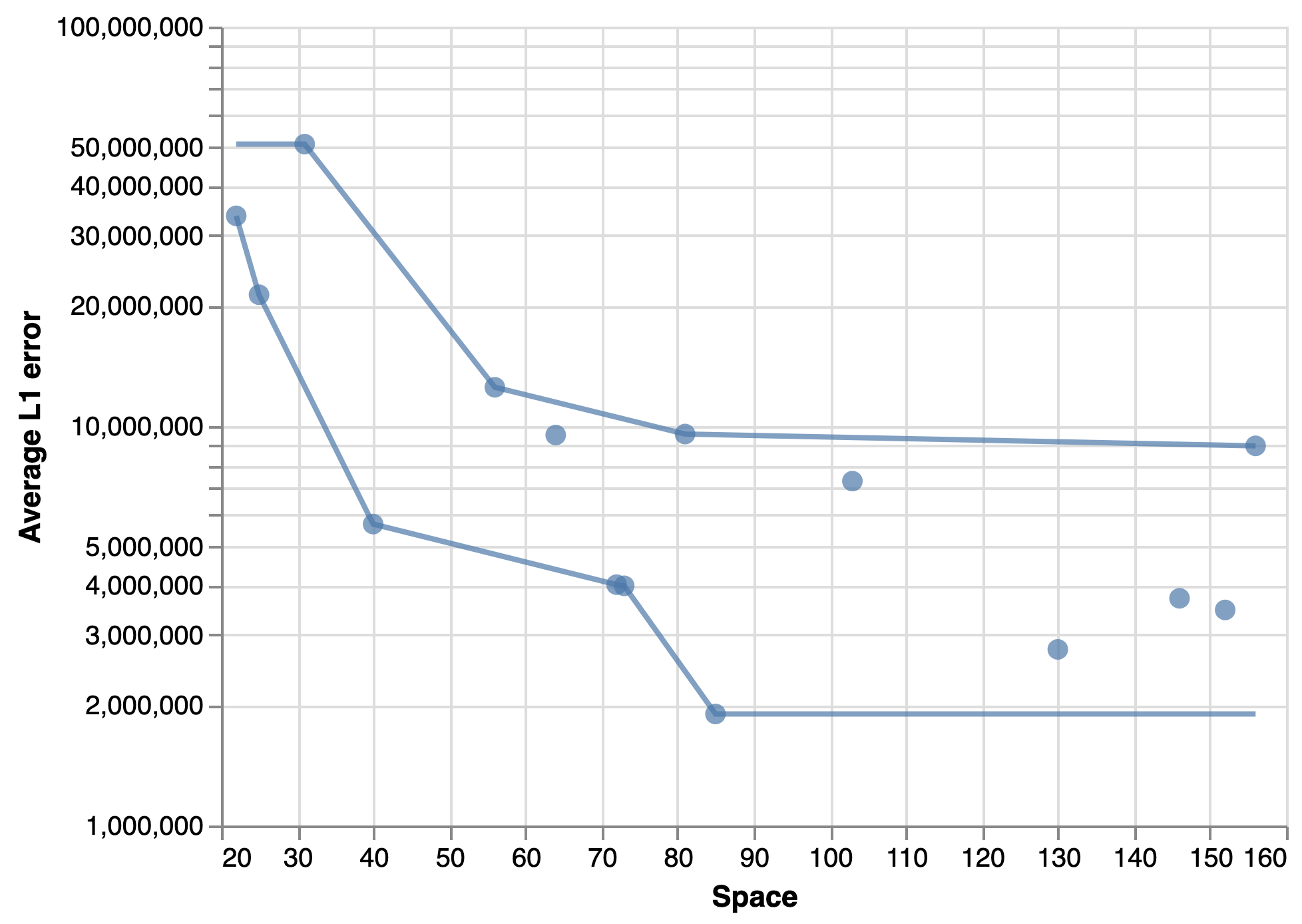}
\caption{An example of the ``envelope'' or hull of the space-error data points for the \algorithmname{} with $t = 2$ (\texttt{books} dataset, random order). It shows the range of errors that a user can expect from the \algorithmname{} given a certain amount of space.}
\label{fig:hull}	
\end{figure}

\begin{figure*}[t]
\includegraphics[width=\textwidth]{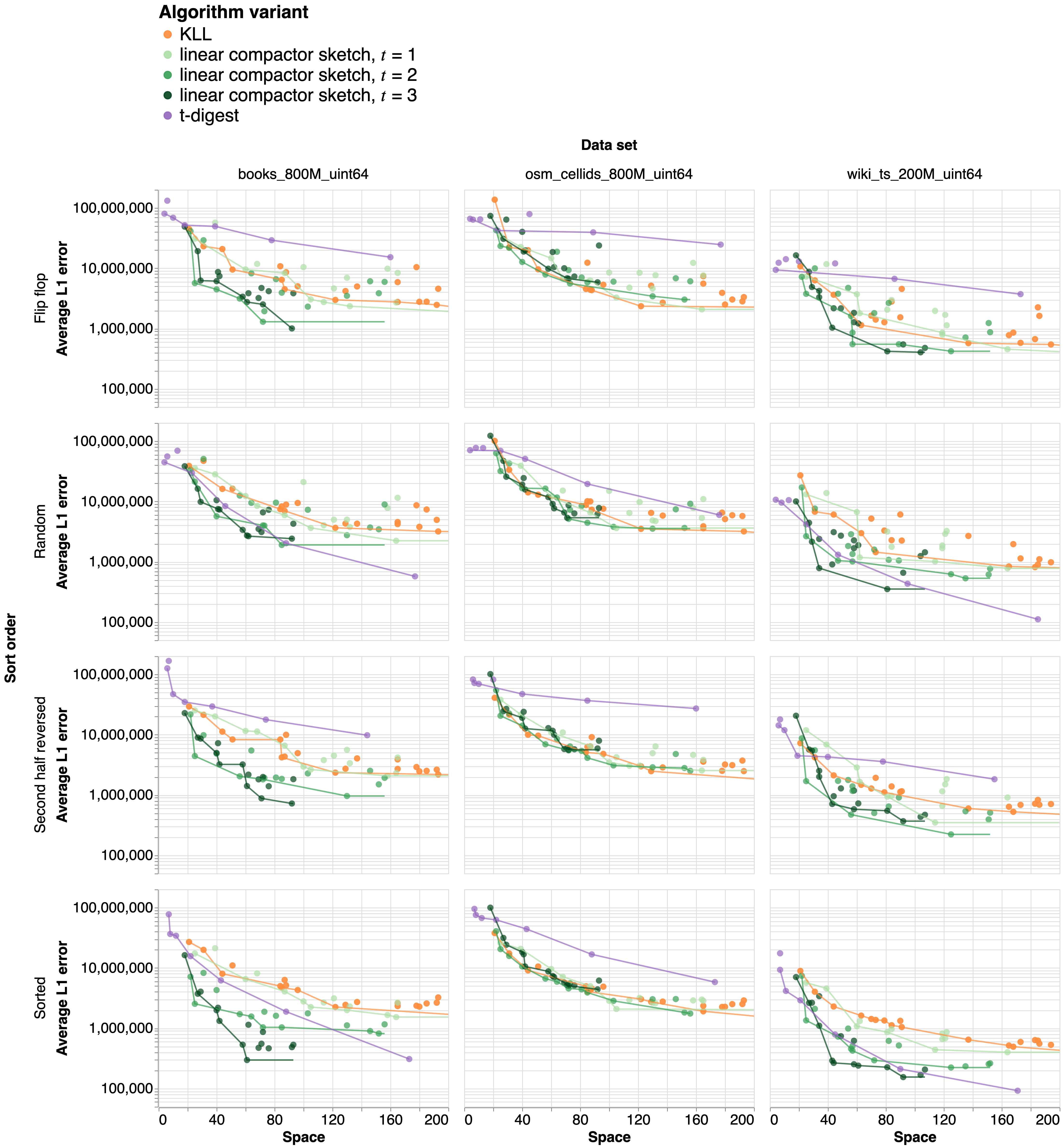}
\caption{Space-error tradeoff curves for the baselines and \algorithmname{} on three different data sets from SOSD and four different sort order, described above.
Markers indicate individual sketches, while curves indicate the lower frontier of possibilities observed (that is, the lower envelope described above) to highlight the general capabilities of each algorithm. 
A better algorithm has a curve that is further down and to the left, indicating lower error at a given amount of space.
}
\label{fig:space-error-tradeoff}
\end{figure*}

\begin{figure*}[t]
\includegraphics[width=\textwidth]{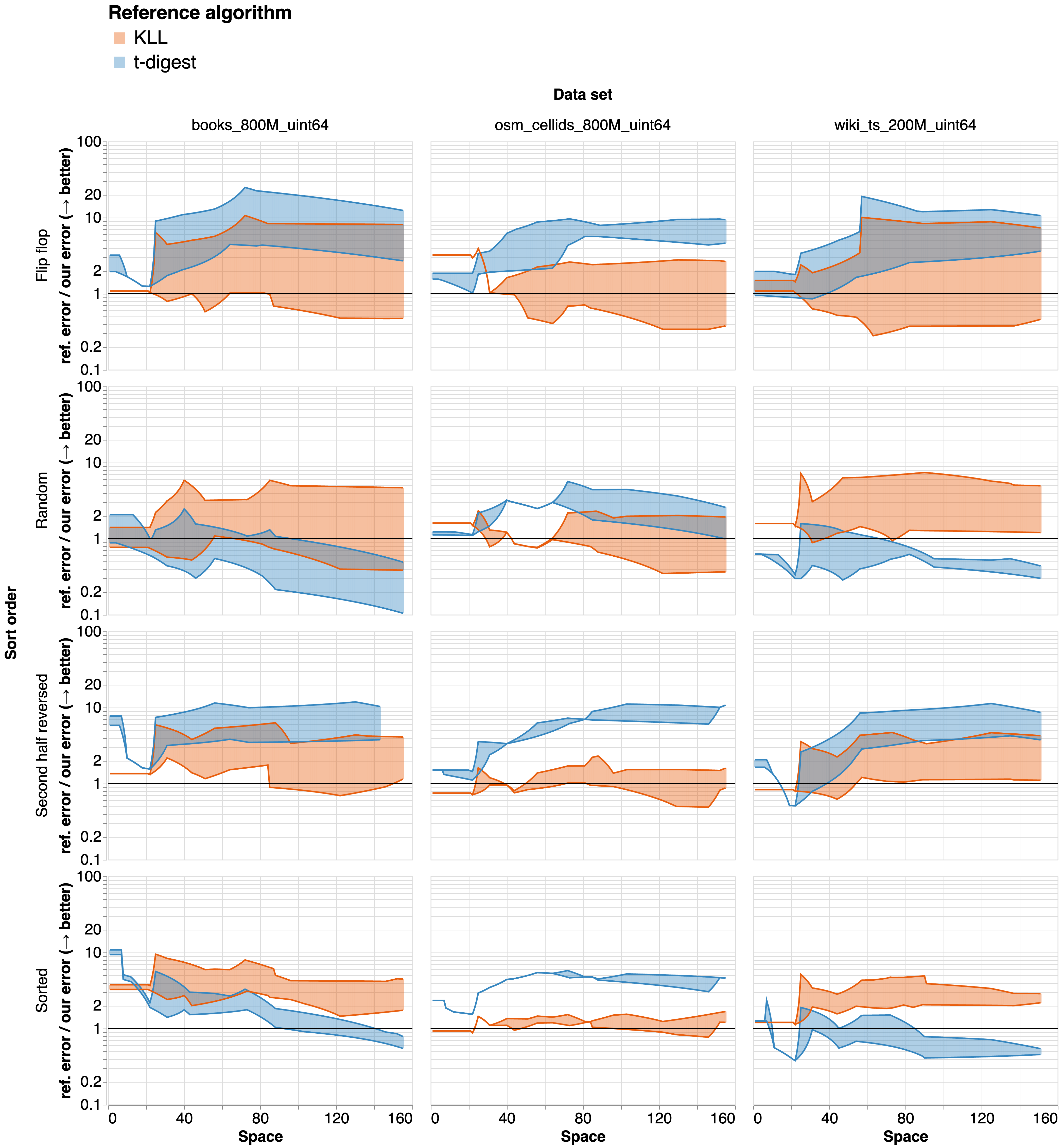}
\caption{Hulls representing the possible ratios of the error achieved by a reference algorithm (either KLL or t-digest) to the error achieved by the \algorithmname{} with $t = 2$. The filled-in area represents the area ratios consistent with the experiments in Figure~\ref{fig:space-error-tradeoff}.
We see that the \algorithmname{} always achieve error no worse than $3\times$ that of KLL, while often achieving and error that is competitive with---and sometimes much lower than---that achieved by the t-digest.
}
\label{fig:error-ratios}	
\end{figure*}

\section*{Acknowledgements}
Justin Y. Chen was supported by a MathWorks Engineering Fellowship, a GIST-MIT Research Collaboration grant, and NSF award CCF-2006798. Justin Y. Chen, Shyam Narayanan, and Sandeep Silwal were supported by NSF Graduate Research Fellowships under Grant No. 1745302. Nicholas Schiefer, Justin Y. Chen, Piotr Indyk,  Shyam Narayanan, and Sandeep Silwal were supported by a Simons Investigator Award. Piotr Indyk was supported by the NSF TRIPODS program (award DMS-2022448).

We thank Sylvia H\"{u}rlimann, Jessica Balik, and the anonymous reviewers for their helpful suggestions.
\FloatBarrier

\bibliography{main}

\begin{thebibliography}{10}

\bibitem{folly}
Folly: Facebook open-source library.
\newblock \url{https://github.com/facebook/folly}.

\bibitem{cormode}
Graham Cormode, Abhinav Mishra, Joseph Ross, and Pavel Vesel{\`y}.
\newblock Theory meets practice at the median: a worst case comparison of
  relative error quantile algorithms.
\newblock In {\em Proceedings of the 27th ACM SIGKDD Conference on Knowledge
  Discovery \& Data Mining}, pages 2722--2731, 2021.

\bibitem{dunning}
Ted Dunning.
\newblock The t-digest: Efficient estimates of distributions.
\newblock {\em Software Impacts}, 7:100049, 2021.

\bibitem{greenwald}
Michael Greenwald and Sanjeev Khanna.
\newblock Space-efficient online computation of quantile summaries.
\newblock {\em ACM SIGMOD Record}, 30(2):58--66, 2001.

\bibitem{ivkin}
Nikita Ivkin, Edo Liberty, Kevin Lang, Zohar Karnin, and Vladimir Braverman.
\newblock Streaming quantiles algorithms with small space and update time.
\newblock {\em Sensors}, 22(24), 2022.

\bibitem{jagadish}
Hosagrahar~Visvesvaraya Jagadish, Nick Koudas, S~Muthukrishnan, Viswanath
  Poosala, Kenneth~C Sevcik, and Torsten Suel.
\newblock Optimal histograms with quality guarantees.
\newblock In {\em VLDB}, volume~98, pages 24--27, 1998.

\bibitem{kll}
Zohar Karnin, Kevin Lang, and Edo Liberty.
\newblock Optimal quantile approximation in streams.
\newblock In {\em 2016 IEEE 57th Annual Symposium on Foundations of Computer
  Science (FOCS)}, pages 71--78, 2016.

\bibitem{kipf}
Andreas Kipf, Ryan Marcus, Alexander van Renen, Mihail Stoian, Alfons Kemper,
  Tim Kraska, and Thomas Neumann.
\newblock Sosd: A benchmark for learned indexes.
\newblock {\em NeurIPS Workshop on Machine Learning for Systems}, 2019.

\bibitem{kraska}
Tim Kraska, Alex Beutel, Ed~H. Chi, Jeffrey Dean, and Neoklis Polyzotis.
\newblock The case for learned index structures.
\newblock In {\em Proceedings of the 2018 International Conference on
  Management of Data}, SIGMOD '18, page 489–504, New York, NY, USA, 2018.
  Association for Computing Machinery.

\bibitem{manku}
Gurmeet~Singh Manku, Sridhar Rajagopalan, and Bruce~G Lindsay.
\newblock Approximate medians and other quantiles in one pass and with limited
  memory.
\newblock {\em ACM SIGMOD Record}, 27(2):426--435, 1998.

\bibitem{marcus}
Ryan Marcus, Andreas Kipf, Alexander van Renen, Mihail Stoian, Sanchit Misra,
  Alfons Kemper, Thomas Neumann, and Tim Kraska.
\newblock Benchmarking learned indexes.
\newblock {\em Proc. {VLDB} Endow.}, 14(1):1--13, 2020.

\bibitem{mitzenmacher}
Michael Mitzenmacher and Sergei Vassilvitskii.
\newblock {\em Algorithms with Predictions}, page 646–662.
\newblock Cambridge University Press, 2021.

\bibitem{digestible}
Dan Morton.
\newblock digestible: A modern c++ implementation of a merging t-digest data
  structure.
\newblock \url{https://github.com/SpirentOrion/digestible}.

\bibitem{wang2013}
Lu~Wang, Ge~Luo, Ke~Yi, and Graham Cormode.
\newblock Quantiles over data streams: An experimental study.
\newblock In {\em Proceedings of the 2013 ACM SIGMOD International Conference
  on Management of Data}, SIGMOD '13, page 737–748, New York, NY, USA, 2013.
  Association for Computing Machinery.

\end{thebibliography}
\bibliographystyle{plain}

\end{document}